\documentclass[11pt]{article}

\usepackage{amsfonts}
\usepackage{amsmath}
\usepackage{amssymb}
\usepackage{amsthm}
\usepackage{authblk}
\usepackage{booktabs}
\usepackage{enumitem}
\usepackage[left=1in,top=1in,right=1in,bottom=1in]{geometry}
\usepackage{graphicx}
\usepackage[bookmarks, colorlinks, breaklinks]{hyperref}
\usepackage{listings}
\usepackage{bbm}
\usepackage[usenames, dvipsnames]{xcolor}
\usepackage{subcaption}
\usepackage{mathtools}
\usepackage{verbatim}

\hypersetup{
	linkcolor=blue,
	citecolor=blue,
	filecolor=black,
	urlcolor=MidnightBlue
}

\newtheorem{theorem}{Theorem}[section]
\newtheorem{lem}[theorem]{Lemma}

\theoremstyle{remark}
\newtheorem*{rem}{Remark}

\DeclareMathOperator{\RE}{Re}
\DeclareMathOperator{\IM}{Im}

\DeclareMathOperator{\cN}{\mathcal{N}}
\DeclareMathOperator{\rN}{\mathrm{N}}
\DeclareMathOperator{\PP}{\mathbb{P}}

\DeclareMathOperator{\vare}{\varepsilon}
\DeclareMathOperator{\Arg}{Arg}

\DeclareMathOperator{\UU}{U}
\DeclareMathOperator{\Sp}{Sp}
\DeclareMathOperator{\SO}{SO}

\newcommand*\der{\mathop{}\!\mathrm{d}}

\title{\Large On the interim statistics for compact group \\characteristic polynomials and their derivatives}

\author[1]{\normalsize Emma Bailey\thanks{Email: \texttt{e.c.bailey@bristol.ac.uk}}}
\author[2]{\normalsize Sebastian Ortiz\thanks{Email: \texttt{sortiz3@ccny.cuny.edu}}}

\affil[1]{\footnotesize  \it Department of Mathematics, University of Bristol, Bristol, UK}
\affil[2]{\footnotesize  \it Department of Mathematics, City College of New York, CUNY, New York, NY}

\date{\today}

\date{}

\begin{document}
\maketitle

\begin{abstract}
  The Keating-Snaith central limit theorem proves that $\Lambda_N(A)=\log\det(I-A)$, for randomly drawn $A\in \UU(N)$, suitably normalised, tends to a complex Gaussian random variable in the large $N$ limit.  The deviations of the real and imaginary parts of $\Lambda_N(A)$, on the scale of a positive $k$th multiple of the variance, are known to be Gaussian but with a multiplicative perturbation in the form of the $2k$th moment coefficient. Here we study the interpolating regime by allowing $k=k(N)$ for both $\RE(\Lambda_N(A))$ and $\IM(\Lambda_N(A))$.  Additionally our methods apply to the logarithm of the derivative of the characteristic polynomial evaluated at an eigenvalue of $A$.
\end{abstract} 

\section{Introduction}
\subsection{Background}
Let $P_N(U,\theta)=\det(I-Ue^{i\theta})$ be the characteristic polynomial for a unitary matrix $U\in \UU(N)$ evaluated at $e^{i\theta}$.  If one draws $A$ from $\UU(N)$ with respect to the Haar measure, then the central limit theorem of Keating and Snaith~\cite{keasna00a} shows the following convergence in distribution for $\theta\in\mathbb{R}$,
\begin{equation}\label{eq:ks_clt}
  \frac{\log P_N(A,\theta)}{\sqrt{\log N}}\xrightarrow[N\rightarrow\infty]{d}\cN_\mathbb{C}(0,1).
\end{equation}
As usual, we say $\mathcal{Z}\sim\cN_{\mathbb{C}}(0,1)$ if the real and imaginary parts of $\mathcal{Z}$ are independently distributed as $\cN(0,1/2)$. The convention we take in this paper is the branch of $\log P_N(A,\theta)$ with locally $-\pi/2<\IM\log(1-e^{i(\theta-\theta_j)})\leq \pi/2$, for $\theta_j$ an eigenangle of $A$. The result is independent of $\theta\in\mathbb{R}$ due to the rotational invariance of the Haar measure.

Equivalently, again for fixed $\theta, x\in\mathbb{R}$, writing $\PP\equiv\PP_{Haar}$ for the usual Haar measure on $U(N)$,
\begin{equation}\label{eq:ks_clt_reim} 
  \begin{rcases*}
    \PP\left(\frac{\log|\det(I-Ae^{i\theta})|}{\sqrt{(1/2)\log N}}\leq x\right)\\
    \PP\left(\frac{\Arg\det(I-Ae^{i\theta})}{\sqrt{(1/2)\log N}}\leq x\right)
  \end{rcases*}
  \overset{N\rightarrow\infty}{\longrightarrow}  \int_{-\infty}^x e^{-u^2/2}\frac{\der u}{\sqrt{2\pi}}.
\end{equation}
See also Figure~\ref{fig:ks_clt}. Here and in the following, unless otherwise stated, we assume that $\theta\in\mathbb{R}$.

A natural extension of \eqref{eq:ks_clt} is to consider fluctuations from the Gaussian limit. Informally, a random variable $X_N$ satisfies a \emph{large deviation principle} with speed $a_N$ and rate function ${I:\mathbb{R}\rightarrow\mathbb{R}_{\geq 0}}$ if $\mathbb{P}(X_N>y)$ decays exponentially as $\exp(-a_N I(y))$ for large $N$. 
Observe that since $A\in\UU(N)$ we have $\log|\det(I-Ae^{i\theta})|\leq N\log 2$ for $\theta\in\mathbb{R}$, but the real part is unbounded below.  The imaginary part has a symmetric linear bound in $N$: $|\Arg\det(I-Ae^{i\theta})|\leq \pi N/2$.  Thus, the maximal scaling for the right tail of the real or imaginary parts of the random variable $X_N=\log P_N(A,\theta)/b_N$ is $b_N$ on the order of $N$.

\begin{figure}[hb]
\centering
    \begin{subfigure}[b]{0.48\textwidth}
    \centering
    \includegraphics[width=\textwidth]{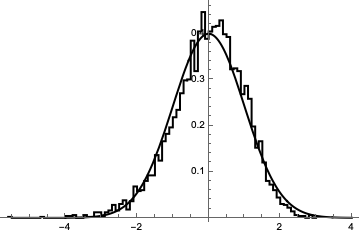}
    \caption{}\label{fig:ks_clt}
    \end{subfigure}
\hfill
    \begin{subfigure}[b]{0.48\textwidth}
    \centering
    \includegraphics[width=\textwidth]{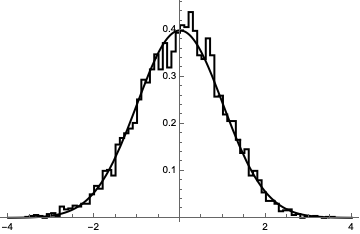}
    \caption{}\label{fig:deriv_clt}
    \end{subfigure}
    \caption{Histograms of (a) $\log|P_N(A,0)|$ and (b) $\log|P_N'(A,\theta_1)|$, both scaled, for $N=75$, with $5,000$ data points, against the standard Normal distribution (solid line).}
\end{figure}

Deviation principles for both the real and imaginary parts of $\log P_N(A,\theta)$ at all scalings $b_N=\mathcal{O}(N)$ were obtained by Hughes, Keating, and O'Connell~\cite{hugkeaoco01}.  In particular, they showed that for $b_N$ growing slower than the critical scaling $N$, for $\IM\log P_N$ the rate function $I$ is always quadratic.  The real part instead displays an asymmetric quality, whereby the right tail has a quadratic rate function up to the critical $\asymp N$ scaling, but the left features a transition from a quadratic to a linear function.  Explicitly, and most pertinently in the context of this work, they show that for fixed, positive $k$ and $\sqrt{\log N}\ll b_N\ll \log N$
\begin{equation}\label{eq:hko_ldp}
  \lim_{N\rightarrow\infty}\frac{1}{a_N}\log\PP\big(\log|P_N(A,\theta)|/b_N\geq k\big)=-k^2,
\end{equation}
where the speed $a_N$ is given as a scaled Lambert's $W$-function (cf.~\cite{hugkeaoco01} Theorem 3.5). For $b_N=\log N$, the asymptotic growth in $N$ is $a_N\sim \log N$. The same statement holds with the imaginary part of the logarithm replacing the real part in \eqref{eq:hko_ldp}. 

Precise large deviations at the particular scale $b_N\asymp \log N$ were obtained by F\'eray, M\'eliot and Nikeghbali~\cite{fermelnik18}.  For example, they showed that for fixed $k>0$
\begin{align}
  \PP\big(\log|P_N(A,\theta)|>k\log N\big)&=c_k\int_{k\sqrt{2\log N}}^\infty e^{-u^2/2}\frac{\der u}{\sqrt{2\pi}}\left(1+o(1)\right)\label{eq:clt_ld}\\
  &=c_k\frac{e^{-k^2\log N}}{k\sqrt{\pi \log N}}\left(1+o(1)\right)\label{eq:clt_ld_asym}
\end{align}
where $c_k$ is the coefficient of the $2k$th moment of $|P_N(A,\theta)|$, which was explicitly calculated in~\cite{keasna00a} to be
\begin{equation}\label{eq:mom_coeff}
  c_k = \lim_{N\rightarrow\infty}\frac{\mathbb{E}\left[|P_N(A,\theta)|^{2k}\right]}{N^{k^2}}=\frac{\mathcal{G}^2(k+1)}{\mathcal{G}(2k+1)}
\end{equation}
where $\mathcal{G}(z)$ is the Barnes $G$-function. 
Therefore~\eqref{eq:clt_ld_asym} captures not only the exponential part of the Gaussian decay from \eqref{eq:hko_ldp}, but the full Gaussian right tail together with a multiplicative perturbation in the form of $c_k$.

Under the, now well-established, conjectural dictionary relating statistical properties of characteristic polynomials to those of the Riemann zeta function, the same questions can be asked on the number theoretic side. The Riemann zeta function is defined as
\begin{equation}\label{rzf}
  \zeta(s)=\sum_{n\geq1}\frac{1}{n^s}=\prod_{p}\Big(1-\frac{1}{p^s}\Big)^{-1}
\end{equation}
for $\RE(s)>1$ and by analytic continuation otherwise. In \eqref{rzf}, the product is taken over primes $p$. It is clear from \eqref{rzf} that $\zeta(s)\neq 0$ for $\RE(s)>1$. The analytic continuation reveals a \emph{functional equation}, cf.~\cite{titchmarsh86}, from which it is clear both that $\zeta(-2n)=0$ for $n\in\{1,2,3,\dots\}$ (the `trivial zeros') and that otherwise if $\zeta(s)=0$ then $0\leq \RE(s)\leq 1$ (these are the `non-trivial' zeros).  The \emph{Riemann hypothesis} states that all the non-trivial zeros have real part equal to $1/2$.

Although the product form in \eqref{rzf} does not hold for $\RE(s)\leq 1$, it transpires that in some sense it does `typically' hold for $\RE(s)=1/2$ (indeed for $\RE(s)\geq 1/2$, cf.~\cite{bohjes32}). Truncating the product in~\eqref{rzf} at some large prime $\mathcal{P}$, we would have
\begin{align}
  \log\zeta(1/2+i t) &\approx \log\prod_{p\leq \mathcal{P}}\left(1-\frac{p^{-it}}{p^{1/2}}\right)^{-1}\nonumber\\
  &=\sum_{p\leq \mathcal{P}} \frac{p^{it}}{p^{1/2}} + \frac{1}{2}\sum_{p\leq \mathcal{P}} \frac{p^{2it}}{p}+\mathcal{O}(1)\label{eq:zeta_model}
\end{align}
after Taylor expanding the logarithm.  As $p^{it}$ is a value on the unit circle for each $p$, it is reasonable to think one could model this contribution by a sequence of random variables taking values uniformly on the unit circle, $(\chi_p)_{p\text{ prime}}$.  Then, the main contribution in the above decomposition is a sum of independent random variables, which implies that $\log\zeta(1/2+it)$, for typical $t$, has a Gaussian structure.  Indeed, this is the statement of Selberg's central limit theorem~\cite{sel46}: for $\tau$ drawn uniformly from $[T,2T]$,
\begin{equation}\label{eq:selbergclt}
  \frac{\log \zeta(1/2 + i\tau)}{\sqrt{\log\log T}}\xrightarrow[T\rightarrow\infty]{d}\cN_\mathbb{C}(0,1).
\end{equation}
Notice the similarity to \eqref{eq:ks_clt} once the standard identification $\log T\leftrightarrow N$ is made.  Unlike the proof of~\eqref{eq:ks_clt}, the Selberg central limit theorem does not use method of moments (although see~\cite{radsou17} where by working just off the critical line $s=1/2+it$ for $\RE\log\zeta(1/2+i\tau)$ one can proceed using moments).  It is conjectured that the moments are asymptotically, for fixed $k\geq 0$,
\begin{equation}\label{eq:zeta_moments}
  \frac{1}{T}\int_{T}^{2T}|\zeta(1/2+it)|^{2k}dt\sim a_k c_k (\log T)^{k^2}
\end{equation}
as $T\rightarrow\infty$, where $a_k$ is an explicit product over prime numbers and $c_k$ is as in \eqref{eq:mom_coeff}, cf.~\cite{titchmarsh86, ivic91, keasna00a}.  The cases $k=1, 2$ are proven~\cite{harlit18, ing26}; unconditional lower bounds of size $\gg_k (\log T)^{k^2}$ are known as are upper bounds $\ll_k (\log T)^{k^2}$, for $k>2$ conditional on the Riemann hypothesis~\cite{titchmarsh86, ram80b, ram78, ram80c, heabro81, heabro93, sou09, har13b, hearadsou19}.  

Interim right-tail deviations to Selberg's central limit theorem have been established~\cite{rad11, ino19} showing Gaussian decay for the likelihood that $\log|\zeta(1/2+it)|$ exceeds $V$ for the range $\sqrt{\log\log T}\ll V\ll (\log\log T)^{2/3}$. In~\cite{rad11}, by considering the model described in~\eqref{eq:zeta_model}, it is instead conjectured that for $k>0$
\begin{equation}
  \frac{1}{T}\left|\left\{T\leq t\leq 2T: |\zeta(1/2+it)|>(\log T)^k\right\}\right|= a_k c_k \int_{k\sqrt{2\log N}}^\infty e^{-u^2/2}\frac{\der u}{\sqrt{2\pi}}\left(1+o(1)\right)\label{eq:clt_ld_zeta}
\end{equation}
as $T\rightarrow\infty$ (i.e.~the deviation likelihood for $V\asymp \log\log T$). Relating again $\log T$ with matrix size $N$, one sees the similarity to \eqref{eq:clt_ld}, lending further support to \eqref{eq:clt_ld_zeta}. Numerical evidence towards \eqref{eq:clt_ld_zeta} is given in~\cite{aabhr21}, where this large deviation is related to the question of typical local maxima of $|\zeta(1/2+it)|$. In proving the sharp conditional upper bounds of \eqref{eq:zeta_moments}, in~\cite{sou09} and \cite{har13b} the bound $\ll_k (\log T)^{-k^2}$ is established.  In~\cite{argbai23} this is strengthened, unconditionally, to $\ll_k \exp(-k^2\log\log T - (1/2)\log\log\log T)$ assuming $k\in(0,2)$.  See also~\cite{argbai25} where matching (unconditional) lower bounds of the same size are shown for all $k> 0$. These bounds are consistent with \eqref{eq:clt_ld_zeta}. 

For the final part of this introduction, we turn to the derivative of $P_N(A,\theta)$.  Again, there is a related theory for the derivative of $\zeta(1/2+it)$.  Once more, a central limit theorem holds, cf.~Figure~\ref{fig:deriv_clt}.  Write $P'_N(A,\theta_j)$ for $(\der/\der\theta)P_N(A,\theta)$ evaluated at an eigenvalue $e^{i\theta_j}$ of $A$.  By computing the moment generating function
\[\mathbb{E}\left[\frac{1}{N}\sum_{n=1}^N|P_N'(A,\theta_n)|^{2k}\right]=\mathbb{E}\left[|P_N'(A,\theta_1)|^{2k}\right]\]
at finite $N$, Hughes, Keating, and O'Connell~\cite{hugkeaoco00} show
\begin{equation}\label{eq:hko_discreteclt}
  \frac{\RE\log P_N'(A,\theta_1) -S_1(N)}{\sqrt{S_2(N)}}\xrightarrow[N\rightarrow\infty]{d}\cN(0,1).
\end{equation}
The mean and variance are
\begin{align*}
  S_1(N)&=\log N + \gamma - 1 +\mathcal{O}(N^{-1})\\
  S_2(N)&=\frac{1}{2}(\log N + \gamma + 3 -3\zeta(2))+\mathcal{O}(N^{-1}),
\end{align*}
where $\gamma$ is the Euler-Mascheroni constant. They additionally showed large deviation principles, including precise asymptotics of the probability density function for the left tail.  For the right tail, the LDP for $\RE\log (P_N'(A,\theta_1)e^{-S_1})$ matches that of $\RE\log P_N(A,0)$. 

Studying the derivative of $P_N(A,\theta)$ is motivated, in part, again in connection with the Riemann zeta function. Speiser's Theorem~\cite{spe35} gives that the Riemann hypothesis is equivalent to $\zeta'(s)$ having no zeros $\rho$ with $\RE(\rho)\in(0,1/2)$. The `discrete moments' are
\begin{equation}
\label{eq:disc_mom}
J_{k}(T)=\frac{1}{\rN(T)}\sum_{0<\IM(\rho)\leq T}|\zeta'(\rho)|^{2k}
\end{equation}
where $\rho$ is a non-trivial zero of the Riemann zeta function, and $\rN(T)\sim (T/2\pi)\log(T/2\pi e)$ is the count of zeros with $0\leq \IM(\rho)\leq T$.  Bounds on the average $J_k(T)$ (with mollifiers) give lower bounds on the asymptotic proportion of simple non-trivial zeros as well as control over gaps between zeros (e.g.~\cite{conghogon98, buihb13}).

Thanks to the Hadamard factorisation of $\zeta(s)$, at least heuristically it is not hard to argue that $\log|\zeta(1/2+i\tau)|$ and $\log(|\zeta'(1/2+i\tau)|/\log T)$ are essentially the same random variable.  Indeed, under the assumption of the Riemann Hypothesis, Hejhal~\cite{hej89} proved a central limit theorem (cf.~\eqref{eq:selbergclt}) for the latter, for $\tau$ drawn uniformly from $[T, 2T]$ (made unconditional in unpublished work of Selberg~\cite{selberg_unp}). 

In a related but different direction, in~\cite{hej89} a discrete version of the central limit theorem for the derivative  is also established: 
\begin{equation}\label{eq:zeta_discreteclt}
  \frac{1}{\rN(2T)-\rN(T)}\left|\left\{T\leq \IM(\rho)\leq 2T: \frac{\RE\log\zeta'(\rho) - \log\left|\frac{1}{2\pi}\log\frac{\IM(\rho)}{2\pi}\right|}{\sqrt{1/2\log\log T}}\geq x\right\}\right|\xrightarrow[T\rightarrow\infty]{d}\cN(0,1)
\end{equation}
for fixed $x\in\mathbb{R}$, under the assumption of the Riemann hypothesis and an assumption on zero-spacing (for example, Montgomery's Pair Correlation conjecture). See also~\cite{cicek21} for an explicit rate of convergence and a discrete analogue for~\eqref{eq:selbergclt}. In \eqref{eq:zeta_discreteclt}, we write $\rho$ for a non-trivial zero of $\zeta(s)$. Making again the connection $\log T$ with matrix size $N$, one sees the connection between~\eqref{eq:hko_discreteclt} and \eqref{eq:zeta_discreteclt}.

\subsection{Results}
Following~\eqref{eq:hko_ldp} and~\eqref{eq:clt_ld_asym}, naturally the question arises of identifying the regime in which $c_k$ appears: controlling the scaling $\log |P_N(A,\theta)|/b_N$ between $b_N\asymp \sqrt{\log N}$ (Gaussian CLT) to $b_N\asymp \log N$ (the deviation regime where the multiplicative coefficient~\eqref{eq:mom_coeff} appears). This is addressed by the following result. 

\begin{theorem}\label{thm1}
  Let $\theta\in\mathbb{R}$ and draw $A\in \UU(N)$ with respect to Haar measure. Write $\rho_N$ for the probability density function for $\RE\log P_N(A,\theta)/\sqrt{Q_2(N)}$ where
  \[Q_2(N) = \frac{1}{2}(\log N + \gamma + 1) + \mathcal{O}\Big(N^{-2}\Big)\]
  is the second cumulant of $\RE\log P_N(A,\theta)$. Set, for $\alpha\geq0$, and $\kappa>0$ fixed
  \begin{align*}
    x\equiv x(N;\vare)
    &= \kappa \sqrt{\frac{(\log N)^{1+\vare}}{Q_2(N)}}\\[1em]
    \vare\equiv \vare(N;\alpha)&= 1-(\log\log N)^{-\alpha}.
  \end{align*}
  Then as $N\rightarrow\infty$, writing $n=\log\log N$ for brevity,
  \begin{align}
    \rho_N(x(N;1-n^{-\alpha})) \sim\frac{1}{\sqrt{2\pi}}\cdot e^{-\kappa^2\exp(n-n^{1-\alpha})}\cdot
    \begin{cases}
      c_\kappa,&\alpha >1\\[1em]
      c_{\frac{\kappa}{\sqrt{e}}},&\alpha =1\\[1em]
      1,&\alpha\in[0,1).
    \end{cases}\label{eq:thm1_density_alpha}
  \end{align}
  The coefficient $c_\kappa$ is given in \eqref{eq:mom_coeff}.

  Finally, setting $\alpha\equiv\alpha(N)=1\pm \frac{1}{\log N\log\log\log N}$ resolves the transitions across $\alpha =1$ in \eqref{eq:thm1_density_alpha}. 
\end{theorem}

Since
\[x(N;\vare)
\sim \sqrt{2}\kappa \exp\Big(\frac{1}{2}\big(1-n^{-\alpha}\big)n\Big),\]
the first exponential term in \eqref{eq:thm1_density_alpha} corresponds to the Gaussian decay.  Thus, when $\alpha\in[0,1)$, this shows that there is no additional multiplicative factor and the central limit theorem extends to this range of $x(N;1-(\log\log N)^{-\alpha})$.  For $\alpha\geq 1$, however, there is a multiplicative shift equal to the moment coefficient \eqref{eq:mom_coeff}, matching the result found in~\cite{fermelnik18} (see also~\cite{aabhr21} and \eqref{eq:clt_ld_zeta}). Allowing $\alpha$ to vary captures the point at which this coefficient enters the tail.    

The equivalent statement holds also for the imaginary part of the logarithm.
\begin{theorem}\label{thm2}
  Let $\theta\in\mathbb{R}$ and draw $A\in \UU(N)$ with respect to Haar measure. Write $\nu_N$ for the probability density function for $\IM\log P_N(A,\theta)/\sqrt{R_2(N)}$ where
  \[R_2(N) = Q_2(N)=\frac{1}{2}(\log N + \gamma + 1) + \mathcal{O}\Big(N^{-2}\Big)\]
  is the second cumulant of $\IM \log P_N(A,\theta)$.  Then with $x\equiv x(N;\vare), \alpha, \vare, \kappa, n$ as in the statement of Theorem~\ref{thm1}
  \begin{align}
    \nu_N(x(N;1-n^{-\alpha})) \sim\frac{1}{\sqrt{2\pi}}\cdot  e^{-\kappa^2\exp(n-n^{1-\alpha})}\cdot
    \begin{cases}
      d_\kappa,&\alpha >1\\[1em]
      d_{\frac{\kappa}{\sqrt{e}}},&\alpha=1\\[1em]
      1,&\alpha\in[0,1).
    \end{cases}\label{eq:thm_density_alpha_im}
  \end{align}
  Above
  \[d_\kappa = |\mathcal{G}(1+i\kappa)|^2\]
  is the coefficient of the $2\kappa$th exponential moment of $\Arg P_N(A,\theta)$.  Once again setting $\alpha\equiv\alpha(N)=1\pm \frac{1}{\log N\log\log\log N}$ resolves the discontinuity in \eqref{eq:thm_density_alpha_im}. 
\end{theorem}

As in Theorem~\ref{thm1}, \eqref{eq:thm_density_alpha_im} captures the point at which the multiplicative perturbation enters the density expression.  In this case, unsurprisingly (see also~\cite{aabhr21}) the coefficient corresponds to a moment of $\exp(\IM\log P_N(A,\theta))$.

Turning now to the derivative of the characteristic polynomial, in the following we recall that we write $P_N'(A,\theta_1)$ for $(\der/\der \theta)P_N(A,\theta)$ evaluated at an eigenvalue $e^{i\theta_1}$ of $A$. The corresponding central limit theorem (cf.~\eqref{eq:hko_discreteclt}) was established in~\cite{hugkeaoco00}.

\begin{theorem}\label{thm3}
   Draw $A\in \UU(N)$ with respect to Haar measure and  write $\varsigma_N$ for the probability density function of  $(\RE\log P_N'(A,\theta_1)-S_1(N))/\sqrt{S_2(N)}$ where
  \begin{align*}
    S_1(N)&=\log N + \gamma - 1 +  \mathcal{O}(N^{-1})\\
    S_2(N)&=\frac{1}{2}(\log N + \gamma + 3 -3\zeta(2))+ \mathcal{O}(N^{-1})
  \end{align*}
  are the first and second cumulants of $\RE\log P_N'(A,\theta_1)$, and $\zeta(s)$ is the Riemann zeta function. 
  Set for $\kappa>0$
  \begin{equation*}
    x\equiv x(N;\vare)=\kappa\sqrt{\frac{(\log N)^{1+\vare}}{S_2(N)}}
  \end{equation*}
  and as before $\vare(N;\alpha)= 1-(\log\log N)^{-\alpha}$ for $\alpha\geq 0$, and $n=\log\log N$.  Then 
  \begin{align}
    \varsigma_N(x(N;1-n^{-\alpha})) \sim\frac{1}{\sqrt{2\pi}}\cdot e^{-\kappa^2\exp(n-n^{1-\alpha})}
    \begin{cases}
      f_\kappa,&\alpha >1\\[1em]
      f_{\frac{\kappa}{\sqrt{e}}},&\alpha=1\\[1em]
      1,&\alpha\in[0,1).
    \end{cases}\label{eq:thm3_derdensity_alpha_re}
  \end{align}
  Above, 
  \[f_\kappa = \frac{\mathcal{G}^2(\kappa+2)}{\mathcal{G}(2\kappa+3)}\]
  is the coefficient of the $2\kappa$th moment of $|P_N'(A,\theta_1)|e^{-S_1}$.  As in Theorems~\ref{thm1} and~\ref{thm2}, setting $\alpha\equiv\alpha(N)=1\pm \frac{1}{\log N\log\log\log N}$ resolves the discontinuity.
\end{theorem}

It is possible to also prove a central limit theorem for the imaginary part of the logarithm for the derivative of $P_N(A,\theta)$, evaluated at an eigenvalue of $A$. This complements the result~\eqref{eq:hko_discreteclt} of~\cite{hugkeaoco00}.  As is to be expected, the imaginary part displays a symmetry lacking in the case of the real part. 
\begin{theorem}\label{thm4}
  Draw $A\in\UU(N)$ with respect to Haar measure and write $e^{i\theta_1}$ for an eigenvalue of $A$. Denote by $(T_j(N))_{j\geq 1}$ the sequence of cumulants of $\IM\log P_N'(A,\theta_1)$.  Then as $N\rightarrow\infty$ the following convergence in law holds
  \[\frac{\IM\log P_N'(A,\theta_1)-T_1(N)}{\sqrt{T_2(N)}}\rightarrow \mathcal{N}(0,1)\]
  and $T_1(N)=-\pi/2$, $T_2(N)=Q_2(N)\sim \frac{1}{2}\log N$. 
\end{theorem}
Hence it is possible to prove the equivalent result to Theorem~\ref{thm3}. 
\begin{theorem}\label{thm5}
  Draw $A\in\UU(N)$ with respect to Haar measure and write $e^{i\theta_1}$ for an eigenvalue of $A$. Write $\tau_N$ for the probability density function for $(\IM\log P_N'(A,\theta_1)-T_1(N))/\sqrt{T_2(N)}$ where $T_1(N)=-\pi/2$ and
  \[ T_2(N)=\frac{1}{2}(\log N + \gamma + 1-\zeta(2)) + \mathcal{O}\Big(N^{-2}\Big)\]
  are the first and second cumulants of $\IM\log P_N'(A,\theta_1)$. Then with $x\equiv x(N;\vare), \alpha, \kappa, n$ as in the statement of Theorem~\ref{thm1},
  \begin{align}
    \tau_N(x(N;1-n^{-\alpha})) \sim\frac{1}{\sqrt{2\pi}}\cdot e^{-\kappa^2\exp(n-n^{1-\alpha})}
    \begin{cases}
      g_\kappa,&\alpha >1\\[1em]
      g_{\frac{\kappa}{\sqrt{e}}},&\alpha=1\\[1em]
      1,&\alpha\in[0,1).
    \end{cases}\label{eq:thm5_derdensity_alpha_im}
  \end{align}
  Above,
  \[g_\kappa = |\mathcal{G}(1+i\kappa)|^2\]
  is the coefficient of the $2\kappa$th exponential moment of $\Arg P_N'(A,\theta_1)-T_1$.  As in Theorems~\ref{thm1},~\ref{thm2} and~\ref{thm3}, setting $\alpha\equiv\alpha(N)=1\pm \frac{1}{\log N\log\log\log N}$ resolves the discontinuity.
\end{theorem}

\begin{rem}
  In~\cite{hugpea25a}, Hughes and Pearce-Crump conjecture that under the Riemann hypothesis, for $\RE(\kappa)>-3$,
  \[\frac{1}{\rN(T)}\sum_{0<\IM(\rho)\leq T}\zeta'(\rho)^\kappa\sim \frac{1}{\Gamma(\kappa+2)}\left(\log\frac{T}{2\pi}\right)^{\kappa}\]
  (as well as expressions for higher and mixed derivatives).  This differs from \eqref{eq:disc_mom} due to the absence of the absolute value.  Justification for the conjecture is via the `Hybrid-model'~\cite{gonhugkea07, buigonmil15} and comparison with the equivalent random matrix computation:
  \begin{equation}\label{eq:disc_non_abs_rmt}
    \mathbb{E}\left[P_N'(A,\theta_1)^{\kappa}\right]\sim \frac{(-i)^\kappa}{\Gamma(\kappa+2)}N^{\kappa}.
  \end{equation}
  The results (Theorems~\ref{thm3}--~\ref{thm5}) are in line with this conjecture.  Written another way, \eqref{eq:disc_non_abs_rmt} concerns
  \[\mathbb{E}\left[P_N'(A,\theta_1)^{\kappa}\right]=\mathbb{E}\left[e^{\kappa\left(\RE\log P_N'(A,\theta_1)+i\IM\log P'_N(A,\theta_1)\right)}\right]\sim \mathbb{E}\left[e^{\kappa \RE\log P_N'(A,\theta_1)}\right]\mathbb{E}\left[e^{i\kappa\IM\log P'_N(A,\theta_1)}\right] .\]
  It is possible to extend the analysis from the proofs of Theorem~\ref{thm3}--\ref{thm5} (effectively as in~\cite{keasna00a} for the non-derivative case) to show that indeed both $\RE\log P_N'(A,\theta_1)$ and $\IM\log P_N'(A,\theta_1)$ are independent in the limit, justifying the expectation splitting. Then, the right hand side approaches the product of the moment generating functions of two Gaussian random variables (cf.~Theorem~\ref{thm3} and Theorem~\ref{thm4}).  The linear growth in $N$ comes from the shifted mean of $\RE\log P_N'(A,\theta_1)$ and the quadratic components cancel between the real and imaginary parts. 
\end{rem}

\begin{rem}
  There are many natural extensions one can consider.  For example, the other unitary ensembles $\operatorname{C}$$\beta$$\operatorname{E}$ for $\beta=1$ and $4$ can be handled in much the same way as described in Section~\ref{sec:proofs} for $\beta=2$.  Indeed, the central limit theorems for $\RE\log P_{N,\beta}(A,\theta)$ and $\RE\log P_{N,\beta}'(A, \theta_1)$ (where $A$ is drawn from the $\operatorname{C}$$\beta$$\operatorname{E}$) for both $\beta=1$ and $\beta=4$ have been proved~\cite{keasna00a, hugkeaoco00}.  The density function expressions necessary to derive the interim regimes can be written down similarly, cf.~Lemma~\ref{lemma1} and~\cite{keasna00a}.

  Alternatively, one could consider $P_N(A,\theta)$ for $A\in\Sp(2N)$ or $A\in \SO(2N)$ averaged over the $\operatorname{CUE}$ ($\beta=2$).  Since matrices from these compact groups have eigenvalues appearing in complex conjugate pairs it is natural to average at the symmetry point $\theta=0$. Once again, central limit theorems for $\log P_N(A,0)$ are known~\cite{keasna00b} for symplectic or special orthogonal $A$, with connections to different number theoretic averages. The results described within readily extend to these cases additionally.  
\end{rem}

\begin{rem}
Finally, we remark that the Theorems~\ref{thm1}--\ref{thm5} extend to number theoretic conjectures for $\zeta(1/2+i t)$ and its derivative. For example, taking $\kappa>0, t=\log\log\log T$, analogously defining $x$, and writing $\tilde{\rho}_N$ for the density function of $\RE\log\zeta(1/2+i \tau)/\sqrt{(1/2)\sum_{p\leq T}p^{-1}}$ where $\tau$ is uniform from $[T,2T]$,  we would expect
\begin{equation}\label{conj:zeta}
\tilde\rho_N(x(\log T; 1-t^{-\alpha})) \sim \frac{1}{\sqrt{2\pi}}e^{-\kappa^2\exp(t-t^{1-\alpha})}
\begin{cases}
    a_\kappa c_\kappa,&\alpha>1\\
    a_{\frac{\kappa}{\sqrt{e}}}c_{\frac{\kappa}{\sqrt{e}}},&\alpha=1\\
    1,&\alpha\in[0,1),
\end{cases} 
\end{equation}
where $a_\kappa c_\kappa$ is the coefficient of the $2\kappa$th moment of $|\zeta(1/2+i \tau)|$, see \eqref{eq:zeta_moments}. The conjecture~\eqref{eq:clt_ld_zeta} aligns with the regime $\alpha>1$. 
\end{rem}

\section{Proofs}\label{sec:proofs}

In this section we prove Theorems~\ref{thm1}--\ref{thm5}.  The idea of the proofs is to manipulate explicit formulae for the probability density functions of the appropriate random variables.  For Theorems~\ref{thm1} and~\ref{thm2}, these were computed in~\cite{keasna00a}.  For Theorem~\ref{thm3}, the formula is developed here following an explicit formula for the moments from~\cite{hugkeaoco01}. For Theorem~\ref{thm4} and~\ref{thm5}, the exponential moments of $\IM \log P_N'(A,\theta_1)$ are calculated, and hence the formula for the density function follows.  This establishes both the central limit theorem as well as the interpolating statistics. 

\subsection{Proof of Theorem~\ref{thm1}}\label{sec:proofthm1}

In \cite{keasna00a}, the probability density function of $\RE\log P_N(A,\theta)$ is found by inverting the associated generating function
\begin{equation}\label{eq:re_mgf}
  \mathbb{E}\left[e^{s\RE\log P_N(A,\theta)}\right] = \prod_{j=1}^N \frac{\Gamma(j)\Gamma(j+s)}{\Gamma^2(j+s/2)},
\end{equation}
valid for $\RE(s)\geq1$. Related are the cumulants $(Q_j(N))_{j\geq 1}$, defined via
\[\log \mathbb{E}\left[e^{s\RE\log P_N(A,\theta)}\right] = \sum_{j\geq 1}\frac{Q_j(N)}{j!}s^j,\]
so
\begin{align*}
  Q_1(N)&=0\\
  Q_2(N)&=\frac{1}{2}\log N + \frac{1}{2}(\gamma+1) + \mathcal{O}\Big(\frac{1}{N^2}\Big)\\
  Q_j(N)&=(-1)^j\frac{2^{j-1}-1}{2^{j-1}}\Gamma(j)\zeta(j-1) + \mathcal{O}\Big(\frac{1}{N^{j-2}}\Big)\qquad j\geq 3.
\end{align*}
These imply that as $N\rightarrow\infty$, $\RE\log P_N(A,\theta)/\sqrt{Q_2(N)}$ will satisfy a central limit theorem. Write $\rho_N(x)$ for its probability density function, which takes the form (cf. Eq.~(53) of~\cite{keasna00a})
\begin{align}
  \rho_N(x)=
  \frac{1}{\sqrt{2\pi}}e^{-x^2/2}\left(1+\sum_{m\geq3} A_m(N)\left(\frac{i}{\sqrt{Q_2(N)}}\right)^m\sum_{p=0}^m\binom{m}{p}\mathcal{E}(m,p)\;(-ix)^p
  \right)
  \label{eq:rho}, 
\end{align}
where the factors $A_m(N)$ are defined via\footnote{See also~Eq.~(52) of~\cite{keasna00a} or Appendix A of~\cite{aabhr21}.}
\begin{equation}\label{eq:comb_def_an}
1 + \sum_{m\geq 3}A_m(N) u^m = \exp\left(\sum_{m\geq 3}\frac{Q_m(N)}{m!}u^m\right)
\end{equation}
and $\mathcal{E}(m,p)$ is given by
\begin{equation}\label{eq:ks_efun}
  \mathcal{E}(m,p)=
  \begin{cases}
    (m-p-1)!! &\text{$m-p$ even}\\
    0 &\text{$m-p$ odd}.
  \end{cases}
\end{equation}
Importantly, $(A_m)_{m\geq 3}$ involves only a finite number of the cumulants of index at least $3$, so for example $A_3 = Q_3/3!=-\pi^2/24 + \mathcal{O}(1/N)$ and for all $m\geq 3$, $A_m(N)=\mathcal{O}(1)$. 

As in the statement of the theorem, take
\begin{equation}\label{eq:x_val}
  x\equiv x(N;\vare)=\kappa \sqrt{\frac{(\log N)^{1+\vare}}{Q_2(N)}}
\end{equation}
and consider briefly $\vare$ fixed. Then $x(N;0)\sim \sqrt{2}\kappa$ corresponds to the \emph{central limit theorem} regime~\eqref{eq:ks_clt_reim}, and $x(N;1)\sim \kappa\sqrt{2\log N}$ to the large deviation regime \eqref{eq:clt_ld}. The role of $\vare\geq0$ hence can interpolate between these two regimes. We will hence choose $\vare$ to depend on $N$ in the following way. Define for $\alpha\in\mathbb{R}$ fixed
\[
\vare\equiv\vare(N;\alpha) = 1 - (\log\log N)^{-\alpha}.
\]
Then for $\alpha>0$ we have, $\vare(N;\alpha)\rightarrow 1$ and 
\begin{align*}
  x(N;\vare)
  &=\kappa\exp\left(\frac{1}{2}\log\frac{(\log N)^{1+\vare}}{Q_2}\right) \\
  &\sim \kappa \exp\left(\frac{1}{2}\log(2\log N) -\frac{1}{2}(\log\log N)^{1-\alpha}\right).
\end{align*}
If therefore $\alpha>1$ then $x\sim \kappa\sqrt{2\log N}$. If $\alpha \in (0,1]$ then $x=o(\sqrt{\log N})$. Comparatively, if $\alpha=0$ then $\vare=0$ and $x\sim\sqrt{2}\kappa$.
  
      
Hence, evaluating \eqref{eq:rho} at \eqref{eq:x_val}, using that the sum over $p$ is dominated at $p=m$, and writing $n=\log\log N$ we find
\begin{align}
  \rho_N(x(N;\vare))
    &\sim
    \frac{1}{\sqrt{2\pi}}\exp\Big(-\frac{1}{2}\kappa^2\exp\Big(\Big(2-n^{-\alpha}\Big)n-\log\Big(\frac{1}{2}(\log N + \gamma + 1)+\mathcal{O}(N^{-2})\Big)\Big)\Big)\nonumber\\
    &\qquad\qquad\qquad\cdot\left(1+\sum_{m\geq 3} A_m\left(\kappa\frac{(\log N)^{\frac{1+\vare}{2}}}{\frac{1}{2}\log N}\right)^m\right)\nonumber\\
    &\sim
    \frac{1}{\sqrt{2\pi}}\exp\left(-\kappa^2\exp\Big(n-n^{1-\alpha}\Big)\right)\exp\left(\kappa^2(\gamma+1)e^{-n^{1-\alpha}}\right)\left(1+\sum_{m\geq3} A_m\left(2\kappa e^{-\frac{1}{2}n^{1-\alpha}}\right)^m\right).\label{eq:thm1_density_proof}
\end{align}

To establish \eqref{eq:thm1_density_alpha}, we compare the different ranges of $\alpha$. Firstly, if $\alpha=0$, so $x\sim \sqrt{2}\kappa$, corresponding to the central limit theorem regime, then the terms in the sum all vanish in the limit (recalling that $A_m(N)=\mathcal{O}(1)$ in $N$) and the large-$N$ behaviour of \eqref{eq:thm1_density_proof} is as expected $\sim \exp(-x^2/2)/\sqrt{2\pi}$. Indeed, this statement holds for any $\alpha\in[0,1)$.  Therefore, the multiplicative perturbation found in \eqref{eq:clt_ld_asym} (corresponding to $x\asymp \sqrt{\log N}$) does not appear for $\alpha<1$.

Instead at $\alpha>1$, so $\vare\rightarrow 1$ at a speed faster than $1/\log\log N$, the term $n^{1-\alpha}$ instead vanishes in the limit, yielding
\begin{align*}
  \rho_N(x)&\sim \frac{1}{\sqrt{2\pi}}\exp\left(-\kappa^2(n-n^{1-\alpha})\right)\cdot e^{\kappa^2(\gamma+1)}\left(1+\sum_{m\geq 3}A_m (2\kappa)^m\right).
\end{align*}
We conclude by noting that $c_\kappa =1+\sum_{m\geq 3}A_m (2\kappa)^m$ is the coefficient of the $2\kappa$th moment of $|P_N(A,0)|$. Indeed, by~\cite{keasna00a} (see also~\cite{aabhr21}, (A18)) we have
\[c_\kappa = \lim_{N\rightarrow\infty}\mathbb{E}[|P_N(A,0)|^{2\kappa}]N^{-\kappa^2}=\lim_{N\rightarrow\infty}\exp\left(\kappa^2Q_2(N)-\kappa^2\log N +\sum_{m\geq 3}\frac{Q_m(N)}{m!}(2\kappa)^m\right).\]
This proves the claim after recalling \eqref{eq:comb_def_an}.

At the value $\alpha=1$ (so $\vare=1-n^{-1}$) we see
\[
\rho_N(x(N;1-1/n)) \sim \frac{1}{\sqrt{2\pi}}\exp\left(-\frac{\kappa^2}{e}e^n\right)e^{\frac{\kappa^2}{e}(\gamma+1)}\left(1+\sum_{m\geq 3}A_m \Big(\frac{2k}{\sqrt{e}}\Big)^m\right).
\]
Collectively this shows \eqref{eq:thm1_density_alpha}.  

There is a discontinuity between the ranges of $\alpha$ in \eqref{eq:thm1_density_alpha} either side of $\alpha=1$.  By setting $\alpha\equiv\alpha(N)=1\pm(\log N\log\log\log N)^{-1}$, one can resolve this discontinuity.  Indeed take $\alpha(N) = 1+(\log N\log\log\log N)^{-1}$.  Then
\[n^{1-\alpha} 
=\exp\Big(\frac{1}{\log N}\Big)\]
so as $N\rightarrow\infty$, $\alpha(N)\rightarrow1$ from above and $\exp(-n^{1-\alpha})\rightarrow 1/e$.  Similarly, the discontinuity is smoothed to the left of $\alpha=1$ by $\alpha(N)=1-(\log N\log\log\log N)^{-1}$. This concludes the proof of Theorem~\ref{thm1}.

\subsection{Proof of Theorem~\ref{thm2}}

The proof of Theorem~\ref{thm2} follows almost verbatim to Theorem~\ref{thm1}, so we just highlight the key differences.  The appropriate generating function for the imaginary part is
\begin{equation}\label{eq:mgf_imag}
  \mathbb{E}\left[\left(\frac{P_N(A,\theta)}{\overline{P_N(A,\theta)}}\right)^{s}\right]=\mathbb{E}\left[e^{2is\IM \log P_N(A,\theta)}\right]=\prod_{j=1}^N\frac{\Gamma^2(j)}{\Gamma(j+s)\Gamma(j-s)},
\end{equation}
for $s\in\mathbb{C}$. We recall that the imaginary part of the logarithm is defined to have a jump discontinuity of $\pi$ as $\theta$ crosses an eigenangle $\theta_j$. As for the real part, we need the density function $\nu_N$ for $\IM\log P_N(A,\theta)/\sqrt{R_2(N)}$, found by inverting \eqref{eq:mgf_imag}. It can again be expressed via the cumulants
\[\log \mathbb{E}\left[e^{t\IM\log P_N(A,\theta)}\right] = \sum_{j\geq 1}\frac{R_j(N)}{j!}t^j\]
which satisfy~\cite{keasna00a},
\begin{equation}
  \begin{alignedat}{2}\label{eq:imag_cumulants} 
     R_{2j}(N)&=\frac{(-1)^{j+1}}{2^{2j-1}}\sum_{\ell=1}^N\Psi^{(2j-1)}(\ell)=\frac{(-1)^{j+1}}{2^{2j-1}-1}Q_{2j}(N)\\
     R_{2j-1}(N)&=0
  \end{alignedat}
\end{equation}
for all $j\geq 1$, where $Q_j(N)$ are the cumulants for the real part of the logarithm, see Section~\ref{sec:proofthm1}, and $\Psi^{(n)}(z)=\frac{\der^{n+1}}{\der z^{n+1}}\log\Gamma(z)$ is the Polygamma function. Again this implies the central limit theorem~\eqref{eq:ks_clt_reim} but also emphasises a symmetry that is not present for the real part of the logarithm. Note that $R_2(N)=Q_2(N)$. The density function, $\nu_N$ for $\IM\log P_N(A,\theta)/\sqrt{R_2(N)}$ is then
\begin{equation}\label{eq:imag_density}
  \nu_N(x) = \frac{1}{\sqrt{2\pi}}e^{-x^2/2}\left(1+\sum_{m\geq2}B_{2m}\left(\frac{i}{\sqrt{R_2(N)}}\right)^{2m}\sum_{p=0}^{2m}\binom{2m}{p}\mathcal{E}(2m,p)\;(-ix)^p
  \right)
\end{equation}
where again the sequence $(B_m(N))_{m\geq 4}$ is defined combinatorially via 
\begin{equation}\label{eq:comb_def_bn}
1 + \sum_{m\geq 2}B_{2m}(N) u^{2m} = \exp\left(\sum_{m\geq 2}\frac{R_{2m}(N)}{(2m)!}u^{2m}\right),
\end{equation}
cf.~\cite{keasna00a}, Eq.~(67), and $\mathcal{E}(m,p)$ is given by \eqref{eq:ks_efun}. Once more $B_m(N)=\mathcal{O}(1)$ for $m\geq 4$, for example $B_4(N)=-3\zeta(3)/4+\mathcal{O}(1/N^2)$. Therefore, the analysis of Section~\ref{sec:proofthm1} goes through verbatim, yielding \eqref{eq:thm_density_alpha_im}, where it just remains to establish the different multiplicative coefficient $d_\kappa$. As demonstrated in~\cite{keasna00a},
\begin{equation}\label{eq:imag_coeff}
d_\kappa=\lim_{N\rightarrow\infty}\frac{\mathbb{E}\left[e^{2\kappa\IM\log P_N(A,\theta)}\right]}{N^{\kappa^2}} =|\mathcal{G}(1+i\kappa)|^2.
\end{equation}
Therefore
\begin{align*}
  d_\kappa &= \lim_{N\rightarrow\infty}\exp\left(2\kappa^2 R_2 -\kappa^2\log N + \sum_{m\geq 2}\frac{R_{2m}}{(2m)!}(2\kappa)^{2m}\right)= e^{\kappa^2(\gamma+1)}\left(1 + \sum_{m\geq 2}B_{2m}(2\kappa)^{2m}\right)
\end{align*}
using \eqref{eq:comb_def_bn}.  This is precisely the multiplicative factor one finds evaluating \eqref{eq:imag_density} at $x(N;1-n^{-\alpha})$ for $\alpha>1$. 

\subsection{Proof of Theorem~\ref{thm3}}

Recall that we write $P_N'(A,\phi)=\frac{\der}{\der \theta}P_N(A,\theta)\Big|_{\theta=\phi}$. Therefore, if $\phi=\theta_1$ is an eigenangle of $A$ then
\[\RE\log P_N'(A,\theta_1) = \sum_{j=2}^N\RE\log(1-e^{i(\theta_1-\theta_j)}).\]
Hence, one can interpret $\RE\log P_N'(A,\theta_1)$ as the real part of the logarithm of a characteristic polynomial of an $(N-1)\times(N-1)$  unitary matrix. It is therefore not surprising that the central limit theorem~\eqref{eq:hko_discreteclt} holds. To begin analysing the large deviations, we again turn to the moment generating function~\cite{hugkeaoco00, hughes01}
\begin{equation}\label{eq:mgf_re_discrete}
  \mathbb{E}\left[e^{s\RE\log P_N'(A,\theta_1)}\right]=\frac{\mathcal{G}^2(2+s/2)\mathcal{G}(N+2+s)\mathcal{G}(N)}{\mathcal{G}(3+s)\mathcal{G}(N+1+s/2)N}.
\end{equation}
The cumulants are hence
\begin{align*}
  S_1(N) &= \log N + \gamma - 1 +\mathcal{O}(N^{-1})\\
  S_2(N) &=\frac{1}{2}\log N + \frac{1}{2}(\gamma+3 - 3\zeta(2)) + \mathcal{O}(N^{-1})\\
  S_j(N) &=\mathcal{O}(1),\qquad j\geq 3.
\end{align*}
Once more this is sufficient to establish the central limit theorem. Let $\varsigma_N(x)$ be the density of $(\RE\log P_N'(A,\theta_1)-S_1)/\sqrt{S_2}$.  We first show the equivalent form for $\varsigma_N$ to \eqref{eq:rho} and \eqref{eq:imag_density}.

\begin{lem}\label{lemma1}
  Let $\varsigma_N$ be the probability density function for $(\RE\log P_N'(A,\theta_1)-S_1)/\sqrt{S_2}$, where $(S_m(N))_{m\geq 1}$ is the sequence of the associated cumulants. Then
  \begin{equation}\label{eq:real_der_density}
    \varsigma_N(x) = \frac{1}{\sqrt{2\pi}}e^{-\frac{x^2}{2}}\left(1+\sum_{m\geq 3}C_m(N)\left(\frac{i}{\sqrt{S_2(N)}}\right)^m\sum_{p=0}^m\binom{m}{p}\mathcal{E}(p,m)\;(-ix)^p\right)
  \end{equation}
  where $(C_m(N))_{m\geq 3}$ is a sequence defined combinatorially in terms of the cumulants $(S_m(N))_{m\geq 3}$ via
  \begin{equation}\label{eq:comb_def_cn}
  1+\sum_{m\geq 3}C_m(N) u^m = \exp\left(\sum_{m\geq 3}\frac{S_m(N)}{m!}u^m\right).
  \end{equation}
\end{lem}

\begin{proof}
  The proof closely follows the construction in~\cite{keasna00a}. Define the non-rescaled measure as $\varsigma_N(u) = \sqrt{S_2}\sigma_N(\sqrt{S_2}u+S_1)$. Then
  \begin{align}
    \sigma_N(x)
    &= \frac{1}{2\pi}\int_{\mathbb{R}}e^{-iyx}\mathbb{E}\left[e^{iy\RE\log P_N'(A,\theta_1)}\right]dy\nonumber\\
    &=\frac{1}{\sqrt{2\pi}}\frac{1}{\sqrt{S_2}}e^{-\frac{(x-S_1)^2}{2S_2}} + \frac{1}{2\pi}\int_{\mathbb{R}}e^{iy(S_1-x)-\frac{y^2}{2}S_2}\sum_{\ell\geq 0}\frac{1}{\ell!}\left(\sum_{m\geq 3}\frac{S_m}{m!}(iy)^m\right)^\ell \der y.\label{eq:combinatorial_def_deriv}
  \end{align}
  Therefore
  \begin{align*}
    \varsigma_N(u) &= \frac{1}{\sqrt{2\pi}}e^{-\frac{u^2}{2}} + \frac{1}{2\pi}\int_{\mathbb{R}}e^{iuw - \frac{w^2}{2}}\sum_{m\geq 3}\frac{C_m}{S_2^{m/2}}(iw)^m \der w\\
    &= \frac{1}{\sqrt{2\pi}}e^{-\frac{u^2}{2}}\left(1+\sum_{m\geq 3}C_m\left(\frac{i}{\sqrt{S_2(N)}}\right)^m\sum_{p=0}^m\binom{m}{p}\mathcal{E}(p,m)\;(-ix)^p\right)
  \end{align*}
  upon integration, which exactly matches the structure of~\eqref{eq:rho} (though with different cumulants). The $m$th term of the sequence $(C_m(N))_{m\geq 3}$ is the coefficient of $(iy)^m$ in the expansion of the sum over $\ell$ in \eqref{eq:combinatorial_def_deriv}. 
\end{proof}

Since $S_2(N)\sim Q_2(N)$ and the sum representation for $\varsigma_N$ has exactly the same structure as \eqref{eq:rho}, the analysis in the proof of Theorem~\ref{thm1} goes through verbatim. 
Indeed, from Lemma~\ref{lemma1} one sees that $\RE\log (P_N'(A,\theta_1)/\exp(S_1))$ behaves essentially like $\RE\log P_N(A,0)$.  Therefore, as in the statement of Theorem~\ref{thm3}, we take for $\kappa\geq 0$
\begin{align}
  x(N,\vare)&=\kappa\frac{\sqrt{(\log N)^{1+\vare}}}{\sqrt{S_2}}
  \sim \kappa\frac{e^{\frac{1+\vare}{2}n}}{\sqrt{S_2}}.\label{eq:xval_deriv_re}
\end{align}
We record that
\begin{align*}
  \frac{x(N;\vare)^2}{2}&=\kappa^2\frac{e^{(2-n^{-\alpha})n}}{2S_2}
  \sim \kappa^2 e^{n-n^{1-\alpha}}\Big(1-(\gamma+3-3\zeta(2))e^{-n}\Big).
\end{align*}
Therefore, evaluating Lemma~\ref{lemma1} at \eqref{eq:xval_deriv_re} we arrive at
\begin{align}
  \varsigma_N(x(N;\vare))\sim\frac{e^{-\frac{x^2}{2}}}{\sqrt{2\pi}}\left(1+\sum_{m\geq 3}C_m\left(\frac{x}{\sqrt{S_2}}\right)^m\right)
  &\sim \frac{1}{\sqrt{2\pi}}\exp\left(-\kappa^2 e^{n-n^{1-\alpha}}\right)c_N(\alpha;\kappa)\label{eq:density_re_deriv_alpha}
\end{align}
where
\begin{equation*}
c_N(\alpha;\kappa)=\exp\left(\kappa^2(\gamma+3-3\zeta(2))e^{-n^{1-\alpha}}\right)\left(1+\sum_{m\geq 3}C_m\left(2\kappa e^{-\frac{1}{2}n^{1-\alpha}}\right)^m\right).
\end{equation*}

Now analysing \eqref{eq:density_re_deriv_alpha} for the different ranges of $\alpha$, it is clear that if $\alpha<1$ then $c_N(\alpha;\kappa)\rightarrow 1$ (again as $(C_m(N))_{m\geq 3}$ only involves $(S_m(N))_{m\geq 3}$ the contribution from the sum vanishes in the limit). If $\alpha\geq 1$ then we get a contribution equal to for $\alpha>1$
\begin{equation}\label{eq:deriv_coeff_interim}
  c_N(\alpha;\kappa)\sim \exp\left(\kappa^2(\gamma+3-3\zeta(2))\right)\left(1+\sum_{m\geq 3}C_m\left(2\kappa\right)^m\right)
\end{equation}
and similarly for $\alpha=1$ with $\kappa$ replaced with $\kappa/\sqrt{e}$. It is also simple to check that setting $\alpha=1\pm (\log N\log\log\log N)^{-1}$ smooths the transition either side of $1$. To conclude, we show that \eqref{eq:deriv_coeff_interim} is equal to the coefficient of the $2\kappa$th exponential moment of $\RE\log P_N'(A,\theta_1)-S_1$.  
From~\cite{hugkeaoco00} we have
\begin{equation}\label{eq:imag_deriv_asymp}
\mathbb{E}\left[e^{s\RE\log P_N'(A,\theta_1)}\right]\sim \frac{\mathcal{G}^2(2+s/2)}{\mathcal{G}(s+3)}N^{\frac{s}{4}(s+4)}.
\end{equation}
We thus consider
\begin{align*}
  \mathbb{E}\left[e^{2\kappa (\RE\log P_N'(A,\theta_1)-S_1)}\right]N^{-\kappa^2}
  &= \exp\left(\sum_{m\geq 1}S_m\frac{(2\kappa)^m}{m!} -2\kappa S_1 - \kappa^2\log N\right)\\
  &\sim \exp\left(\kappa^2(\gamma+3 - 3\zeta(2))+\sum_{m\geq 3}S_m\frac{(2\kappa)^m}{m!}\right)
\end{align*}
which matches \eqref{eq:deriv_coeff_interim} using \eqref{eq:comb_def_cn}.

\subsection{Proof of Theorems~\ref{thm4} and~\ref{thm5}}
As in the proof of Theorem~\ref{thm3}, we write
$P_N'(A,\phi)=\frac{\der}{\der \theta}P_N(A,\theta)\Big|_{\theta=\phi}$. 
Therefore, if $\phi=\theta_1$ is an eigenangle of $A$ then
\begin{align}
  \IM\log P_N'(A,\theta_1)
  &= \IM\log\left(-i\prod_{j=2}^N\Big(1-e^{i(\theta_1-\theta_j)}\Big)\right)\nonumber\\
  &=-\frac{\pi}{2}-\sum_{j=2}^N\sum_{\ell\geq 1}\frac{\sin((\theta_1-\theta_j)\ell)}{\ell}.\label{eq:im_part_deriv}
\end{align}

Then, using the explicit form for the Haar measure on $\UU(N)$ we have the following reduction for the exponential moments of $W_N$ (this is, in part, a specialisation of Lemma 1.9 of~\cite{hughes01}.)
\begin{lem}\label{lemma2}
  For $s\in\mathbb{C}$, and $e^{i\theta_1}$ an eigenvalue of $A$, 
  \begin{equation}\label{eq:mgf_imag_deriv}
    \mathbb{E}\left[e^{i s \IM\log P_N'(A,\theta_1)}\right]=\frac{(-i)^s}{N}\mathbb{E}_{(N-1)}\left[e^{2\RE\log \tilde{P}_N(A,0)+i s\IM\log\tilde{P}_N(A,0)}\right]
  \end{equation}
  where the decoration $(N-1)$ on the expectation means average over $\UU(N-1)$ and
  \[\tilde{P}_N(A,\phi)=\prod_{j=2}^N\big(1-e^{i(\theta_j-\phi)}\big).\]
  Further, as $N\rightarrow\infty$, the expectation splits,
  \begin{align}
    \mathbb{E}\left[e^{i s \IM\log P_N'(A,\theta_1)}\right]&\sim\frac{(-i)^s}{N}\mathbb{E}_{(N-1)}\left[e^{2\RE\log \tilde{P}_N(A,0)}\right]\mathbb{E}_{(N-1)}\left[e^{i s\IM\log\tilde{P}_N(A,0)}\right]\nonumber\\
    &\sim (-i)^s\mathcal{G}\Big(1+\frac{s}{2}\Big)\mathcal{G}\Big(1-\frac{s}{2}\Big)N^{-\frac{s^2}{2}}.\label{eq:joint_mgf_asymp}
  \end{align}
\end{lem}

\begin{proof}
 By the explicit form of the Haar measure on $\UU(N)$ and \eqref{eq:im_part_deriv} we have
 \begin{align*}
    \mathbb{E}&\left[e^{i s \IM\log P_N'(A,\theta_1)}\right]\\
    &=\frac{1}{N!}\frac{e^{-i\pi s/2}}{(2\pi)^N}\int_{[0,2\pi]^N}\prod_{m=2}^N\left\{|e^{i\theta_1}-e^{i\theta_m}|^2 e^{-is\sum_{\ell\geq 1}\frac{\sin((\theta_1-\theta_m)\ell)}{\ell}}\right\}\prod_{2\leq j<k\leq N}|e^{i\theta_j}-e^{i\theta_k}|^2\der\theta_1\cdots \der\theta_N\\
    &=\frac{(-i)^s}{2\pi N}\int_0^{2\pi}\mathbb{E}_{(N-1)}\left[\prod_{m=2}^N |1-e^{i\theta_m}|^2 e^{-is\sum_{\ell\geq 1}\frac{\sin(\theta_m\ell)}{\ell}}\right]\der\theta_1
 \end{align*}
 using periodicity.  This establishes \eqref{eq:mgf_imag_deriv}. 
 
 For $A\in\UU(N)$ drawn with Haar measure, as $N\rightarrow\infty$, the random variables $\RE\log P_N(A,\theta)$ and $\IM\log P_N(A,\theta)$ become independent~\cite{keasna00a}, therefore the moment generating function~\eqref{eq:mgf_imag_deriv} will split (see also~\eqref{eq:mgf_imag_deriv_2}).  Further, for such $A$,
 \[\mathbb{E}\left[|P_N(A,\theta)|^2\right] = N+1\]
 so combined with~\eqref{eq:imag_coeff} this shows~\eqref{eq:joint_mgf_asymp}.
\end{proof}

For $B\in\UU(N)$ drawn with Haar measure, the joint moments of $\RE\log P_N(B,\theta)$ and $\IM\log P_N(B,\theta)$ have been computed~\cite{keasna00a, bakfor97, botsil99}, so \eqref{eq:mgf_imag_deriv} is
\begin{align}
  \mathbb{E}\left[e^{i s \IM\log P_N'(A,\theta_1)}\right]
  &=\frac{(-i)^s}{N}\prod_{\ell=1}^{N-1}\frac{\Gamma(\ell)\Gamma(\ell+2)}{\Gamma(\ell+1+s/2)\Gamma(\ell+1-s/2)}.\label{eq:mgf_imag_deriv_2}
\end{align}

Therefore $(T_m(N))_{m\geq 1}$, the cumulants of $W_N(A,\theta_1)$, defined via
\[\log\mathbb{E}\left[e^{t \IM\log P_N'(A,\theta_1)}\right]=\sum_{m\geq 1}\frac{T_m(N)}{m!}t^m,\]
can be found by differentiating the logarithm of \eqref{eq:mgf_imag_deriv_2}:
\begin{align}
  T_m(N)&=\frac{\der^m}{\der t^m}\log\mathbb{E}\left[e^{t \IM\log P_N'(A,\theta_1)}\right]\Big|_{t=0}\nonumber\\
  &=-
  \begin{cases}
    \frac{\pi}{2},&m=1\\
    0,&m>1\text{ odd}\\
    \frac{i^m}{2^{m-1}}\sum_{\ell=1}^{N-1}\Psi^{(m-1)}(\ell+1),&m\text{ even}
  \end{cases}\label{eq:imag_deriv_cumulants}
\end{align}
where $\Psi^{(n)}(z)=\frac{\der^{n+1}}{\der z^{n+1}}\log\Gamma(z)$ is the Polygamma function. Comparing to \eqref{eq:imag_cumulants} we see for $m\geq 1$ fixed
\[T_{2m}(N) = R_{2m}(N)+\frac{(-1)^{m}}{2^{2m-1}}\Psi^{(2m-1)}(1)=R_{2m}(N)+\mathcal{O}(1)\]
since $\Psi^{(n)}(1)=(-1)^{n+1}\Gamma(n+1)\zeta(n+1)$.  In particular 
\[T_2(N)=\frac{1}{2}(\log N + \gamma + 1 - \zeta(2))+\mathcal{O}(N^{-2}).\]
This immediately implies Theorem~\ref{thm4} (with the same speed of convergence as for the non-derivative case). Hence writing $\tau_N(x)$ for the density function of $(W_N(A,\theta_1)-T_1(N))/\sqrt{T_2(N)}$, we have
\begin{equation}\label{eq:imag_density_deriv}
  \tau_N(x) =\frac{1}{\sqrt{2\pi}}e^{-x^2/2}\left(1+\sum_{m\geq2}D_{2m}(N)\left(\frac{i}{\sqrt{T_2(N)}}\right)^{2m}\sum_{p=0}^{2m}\binom{2m}{p}\mathcal{E}(2m,p)\;(-ix)^p
  \right).
\end{equation}
In~\eqref{eq:imag_density_deriv}, as before $D_{2m}(N)$  is defined via 
\begin{equation}\label{eq:comb_def_dn}
1+\sum_{m\geq 2}D_{2m}(N)u^{2m}=\exp\left(\sum_{m\geq 2}\frac{T_{2m}(N)}{(2m)!}u^{2m}\right).
\end{equation}

The analysis of \eqref{eq:imag_density_deriv} for $x(N;\vare)=\kappa \sqrt{(\log N)^{1+\vare}/T_2}$ follows as in the previous three cases since $T_2(N)\sim R_2(N)$.  It remains to show that for $\alpha> 1$ the coefficient $g_\kappa$ in
\[\tau_N(x(N;1-n^{-\alpha}))\sim g_\kappa\cdot \frac{1}{\sqrt{2\pi}}\exp\Big(-\kappa^2\exp(n-n^{1-\alpha})\Big)\]
is the coefficient of the $2\kappa$th exponential moment of $\Arg P_N'(A,\theta_1)-T_1$.
(The $\alpha=1$ case follows similarly.) 

From \eqref{eq:imag_density_deriv}, for $\alpha>1$ we have 
\begin{equation}\label{eq:imag_deriv_coeff_expression}
  g_\kappa = e^{\kappa^2(\gamma+1-\zeta(2))}\left(1+\sum_{m\geq 2}D_{2m}(2\kappa)^{2m}\right),
\end{equation}    
and as usual, since $D_{2m}(N)$ is built solely of cumulants of $\IM\log P_N'(A,\theta_1)$ of index greater than $4$, they are constant in the limit. Similarly, by Lemma~\ref{lemma2} we have 
\begin{align*}
    |\mathcal{G}(1+i\kappa)|^2
    &=\lim_{N\rightarrow\infty}\mathbb{E}[e^{2\kappa(\IM\log P_N'(A,\theta_1)-T_1)}]N^{-\kappa^2}\\
    &=\lim_{N\rightarrow\infty}\exp\Big(-2\kappa T_1 - \kappa^2\log N + 2\kappa T_1 + 2\kappa^2T_2 + \sum_{m\geq 2} \frac{T_{2m}(N)}{(2m)!}(2\kappa)^{2m}\Big)\\
    &=e^{\kappa^2(\gamma+1-\zeta(2))}\cdot \exp\Big(\sum_{m\geq 2} \frac{T_{2m}(N)}{(2m)!}(2\kappa)^{2m}\Big)
\end{align*}
which aligns with the expression \eqref{eq:imag_deriv_coeff_expression} after recalling \eqref{eq:comb_def_dn}.  This concludes the proof of Theorem~\ref{thm5}. 

\bibliographystyle{alpha}
\bibliography{~/Documents/myPapers/ref.bib}

\end{document}